\documentclass[10pt,twocolumn,twoside]{IEEEtran}

\IEEEoverridecommandlockouts                              

\usepackage{graphics} 
\usepackage{epsfig} 
\usepackage{amsmath} 
\usepackage{amssymb}  
 
\usepackage{amsthm}
\usepackage{cite}
\usepackage{color}
\usepackage{enumerate}
\usepackage{graphicx}
\usepackage{mathtools}
\usepackage{cancel}
\usepackage{url}
\usepackage{algorithm}
\usepackage{algpseudocode}
\usepackage{textgreek}
\usepackage{hyperref} 
\usepackage{cleveref}

\crefname{equation}{}{}
\crefname{proposition}{Proposition}{Propositions}

\graphicspath{{figures/}}

\newcommand{\naturals}{\mathbb{N}}
\newcommand{\real}{\mathbb{R}}

\newcommand{\realpos}{\mathbb{R}_{> 0}}
\newcommand{\until}[1]{[#1]}
\newcommand{\map}[3]{#1:#2 \rightarrow #3}

\newcommand{\longthmtitle}[1]{\mbox{}{\textit{(#1):}}}
\newcommand{\setdef}[2]{\{#1 \; | \; #2\}}
\newcommand{\setdefb}[2]{\big\{#1 \; | \; #2\big\}}

\newcommand*{\SetSuchThat}[1][]{} 
\newcommand*{\MvertSets}{%
    \renewcommand*\SetSuchThat[1][]{%
        \mathclose{}%
        \nonscript\;##1\vert\penalty\relpenalty\nonscript\;%
        \mathopen{}%
    }%
}
\MvertSets 
\DeclarePairedDelimiterX \Set [2] {\lbrace}{\rbrace}
    {\,#1\SetSuchThat[\delimsize]#2\,}

\newcommand{\Cc}{\mathcal{C}}

\newcommand{\Ec}{\mathcal{E}}

\newcommand{\Kc}{\mathcal{K}}

\newcommand{\Uc}{\mathcal{U}}
\newcommand{\Sc}{\mathcal{S}}

\newcommand{\Wc}{\mathcal{W}}

\newcommand{\R}{\mathbb{R}}

\newcommand{\Ke}{\Kc^{\rm e}}

\newcommand{\defeq}{\triangleq}

\newtheorem{theorem}{Theorem}
\newtheorem{lemma}{Lemma}
\newtheorem{corollary}{Corollary}
\newtheorem{proposition}{Proposition}

\theoremstyle{definition}
\newtheorem{definition}{Definition}

\newtheorem{remark}{Remark}

\renewcommand{\bf}{\mathbf{f}} 
\newcommand{\bg}{\mathbf{g}}

\newcommand{\bk}{\mathbf{k}}

\newcommand{\bmm}{\mathbf{m}} 

\newcommand{\bu}{\mathbf{u}}
\newcommand{\bv}{\mathbf{v}}

\newcommand{\bx}{\mathbf{x}}
\newcommand{\by}{\mathbf{y}}

\newcommand{\bA}{\mathbf{A}}
\newcommand{\bB}{\mathbf{B}}

\newcommand{\bD}{\mathbf{D}}

\newcommand{\bF}{\mathbf{F}}

\newcommand{\bH}{\mathbf{H}}
\newcommand{\bI}{\mathbf{I}}

\newcommand{\bL}{\mathbf{L}}
\newcommand{\bM}{\mathbf{M}}

\newcommand{\bP}{\mathbf{P}}

\newcommand{\bV}{\mathbf{V}}

\newcommand{\bLambda}{\boldsymbol{\Lambda}}
\newcommand{\balpha}{\boldsymbol{\alpha}}
\newcommand{\bzero}{\mathbf{0}}

\newcommand{\obs}{{\operatorname{o}}}
\newcommand{\des}{{\operatorname{d}}}

\newcommand{\diag}{{\operatorname{diag}}}

\allowdisplaybreaks

\DeclareMathOperator*{\argmin}{argmin}

\title{\LARGE \textbf{Matrix Control Barrier Functions}}

\author{Pio Ong, Yicheng Xu, Ryan M. Bena, Faryar Jabbari, Aaron D. Ames %
\thanks{PO, RB, and AA are with the Department of Mechanical and Civil Engineering, California Institute of Technology, Pasadena, CA 91125, USA, \texttt{\{pioong, ryanbena, ames\}@caltech.edu}.}
\thanks{YX and FJ is with the Department of Mechanical and Aerospace Engineering, University of California, Irvine, CA 92697, USA, \texttt{\{yichex7, fjabbari\}@uci.edu}.}
\thanks{This research was in parts supported by the Technology Innovation Institute.}
}

\begin{document}

\maketitle
\begin{abstract}
    This paper generalizes the control barrier function framework by replacing scalar-valued functions with matrix-valued ones. Specifically, we develop barrier conditions for safe sets defined by matrix inequalities---both semidefinite and indefinite. Matrix inequalities can be used to describe a richer class of safe sets, including nonsmooth ones.   The safety filters constructed from our proposed matrix control barrier functions via semidefinite programming (CBF-SDP) are shown to be continuous. Our matrix formulation naturally provides a continuous safety filter for Boolean-based control barrier functions, notably for disjunctions (OR), without relaxing the safe set. We illustrate the effectiveness of the proposed framework with applications in drone network connectivity maintenance and nonsmooth obstacle avoidance, both in simulations and hardware experiments.
\end{abstract}
\section{Introduction}
Dynamic safety is increasingly recognized as essential in modern control systems alongside stability. While convergence to the equilibrium is important, it is equally critical that safety constraints are respected throughout the entire trajectory. Control barrier functions (CBFs)~\cite{ADA-SC-ME-GN-KS-PT:19,PW-FA:07} have emerged as an effective framework for enforcing safety during system evolution. A key advantage of CBFs is their integration into the safety filters~\cite{TG-MM-AS-PN-EF-ADA:20} as constraints in a quadratic program (QP), which facilitates the integration of multiple safety and stability requirements and enables fast online computation. The success of CBFs coincided with the rise of real-time QP solvers. Today, with advances in semidefinite program (SDP) solvers and computational hardware, it becomes practical to move beyond scalar-valued barrier functions and define richer safe sets via matrix inequalities, motivating the development of \textit{matrix control barrier functions} (MCBF). 

The original CBF framework relies on Nagumo's theorem~(see \cite{MN:42} or~\cite[Ch. 4.2]{FB-SM:07}) to ensure forward invariance of safe sets. To use the result, safety-critical controllers must render the closed-loop dynamics locally Lipschitz. This requirement motivated substantial research into the regularity of safety filters derived from the QP formulation, which is the most common form of CBF controller. For single-constraint QPs, it is simple to guarantee (local) Lipschitz continuity by examining closed-form solutions, while for multiple-constraint QPs, continuity and Lipschitz property have been studied using tools from set-valued analysis and parametric optimization~\cite{BJM-MJP-ADA:13,BJM-MJP-ADA:15,MJ:18,RAF-PVK:96}. Recently, the work~\cite{PM-AA-JC:25} compiles different sufficient conditions, such as the linear independence constraint qualification (LICQ), which guarantees continuity and Lipschitz properties for QP-based safety filters.

In contrast, the modern CBF framework does not rely on Nagumo's theorem. Instead, it requires that the CBF condition hold in a neighborhood of the safe set, enabling safety guarantees from continuous safety filters. This relaxation is particularly important when extending the control formulation beyond QPs. While the most prominent form of CBF-based controllers relies on the QP formulation, the use of other convex optimization classes has also emerged. For instance, measurement-robust CBFs~\cite{RKC-AWS-AJT-TGM-KLB-ADA:21} yield barrier constraints that require second-order cone programs (SOCPs) to formulate control laws. The matrix-valued CBFs proposed in this paper are enforced via SDPs, where Lipschitz continuity guarantees are even more difficult to establish.

Matrix inequalities provide a natural way to describe nonsmooth safe sets, including those defined by Boolean combinations of constraints. Due to their complexity, work on nonsmooth safe sets is relatively sparse. The paper~\cite{PG-JC-ME:17} formalizes a nonsmooth barrier function framework that offers barrier conditions for verifying forward invariance. The work was motivated by the natural nonsmoothness arising from combining multiple CBFs into one via Boolean operations, both for conjunction (AND) and disjunction (OR). Unfortunately, it does not provide a controller construction method for control systems.
Specifically for Boolean-based CBFs, \cite{TGM-ADA:23} proposes using soft-min and soft-max functions to enable the formulation of safety filters with continuity properties. This approach, however, conservatively alters the safe set.

The motivation of this work arises from the problem of connectivity maintenance in multi-robot systems. In such settings, CBFs are especially attractive because they allow the integration of a connectivity constraint alongside other objectives without requiring a complete co-design of the underlying control law. Approaches for connectivity maintenance can be categorized into local and global. Local approaches \cite{MJ-ME:07,MMZ-GJP:05,MDS-JC:09} address the problem by reasoning about the initial network configuration; on the other hand, global approaches \cite{BC-LS:20,PO-BC-LS-JC:23-auto} rely on algebraic graph theory, where the Fiedler eigenvalue~\cite{MF:73,CDG-GFR:01} must remain  positive at all times to ensure connectivity. Unfortunately, eigenvalues are nonsmooth functions of matrices, complicating the design of continuous control laws. Our prior work \cite{PO-BC-LS-JC:23-auto} attempted to resolve this issue by carefully designing the safety filter constraints to ensure continuity, drawing on nonsmooth analysis tools. In contrast, the matrix CBF proposed herein offers a simpler systematic approach to designing safety filters to address these nonsmooth connectivity constraints.

\textbf{Statement of Contributions:} This paper develops the general framework of \textit{matrix control barrier functions}. We formulate safety constraints described by matrix-valued functions rather than scalar ones, investigating both semidefinite and indefinite cases. To this end, we develop barrier conditions that ensure set forward invariance and their counterparts for control systems that ensure control invariance. We propose optimization-based controllers via semidefinite programming (CBF-SDP) and provide the technical analysis of their continuity properties and the bounds they enforce.

Compared to existing approaches, the proposed framework offers several important advantages. Throughout the paper, we investigate the special case with diagonal matrices and discuss how MCBFs can address Boolean compositions of multiple safety constraints. MCBFs directly handle the nonsmooth nature of the problem without resorting to soft-min or soft-max approximations that conservatively relax the original safe set. In particular, the framework elegantly offers the CBF-SDP solution for disjunctive combinations of safety constraints, which was not previously possible. In connectivity maintenance problems, where prior work relied on nested optimization schemes to address nonsmoothness, our approach formulates the problem as a single SDP, simplifying both analysis and implementation. We validate our MCBF framework by applying the proposed results on the connectivity maintenance problem, both in simulation and experiment.

\section{Background on Safety Filter}
\subsection{Notation}
Throughout the paper, we use $\naturals$ and $\real$ for the set of natural and real numbers. For $p\in\naturals$, we denote by $\until{p}$ the set of consecutive numbers $\{1,2,\ldots,p\}$. Given a vector $\bu\in\real^m$, $\|\bu\|$ is its Euclidean norm. Given a matrix $\bM\in\real^{m\times n}$, $M_{ij}$ denotes its $(i,j)$-th entry. We use $\mathbb S^p$ for the space of symmetric real matrices of size $p\times p$. 
We use $\mathbf{1}_p\in\real^p$ for a vector of ones and $\bI_{p\times p}\in\mathbb S^p$ for the identity matrix. Given two matrices $\bA,\bB\in\real^{p\times p}$, $\bA\cdot\bB$ is their Frobenius product. Note that the Frobenius product satisfies the identity: $\bv\bv^\top\cdot \bA = \bv^\top \bA\bv$ for any $\bv\in\real^p$. A function $\map{\alpha}{\real}{\real}$ is of extended class-$\Kc$, denoted $\alpha\in\Ke$, if it is continuous, strictly increasing, and satisfies $\alpha(0)=0$. For a continuously differentiable function $\map{h}{\real^n}{\real}$, its Lie derivative along a vector field $\map{\bf}{\real^n}{\real^n}$ (or along multiple vector fields stacked in a matrix $\map{\bg}{\real^n}{\real^{n\times m}}$) is defined as $L_\bf h(\bx)=\frac{\partial h}{\partial \bx} \bf(\bx)$.

\subsection{Barrier Condition}
Consider the nonlinear autonomous system:
\begin{equation}\label{sys:auto}
    \dot \bx = \bF(\bx)
\end{equation}
where $\bx\in\real^n$ is the system state and $\map{\bF}{\real^n}{\real^n}$ is the system dynamics. For \textit{safety}, we are interested in verifying that the state trajectories $t\rightarrow\bx(t)$ remain inside a safety constraint defined as a sublevel set of the function $\map{\psi}{\real^n}{\real}$:
\begin{equation}\label{eq:safety-constraint}
    \Sc = \setdefb{\bx\in\real^n}{\psi(\bx)\geq 0},
\end{equation}
at all times. In other words, we want the set $\Sc$ to be forward invariant.
\begin{definition}
    \longthmtitle{Forward Invariance}
    \label{def:forward-invariance}
    A set $\Cc\subset\real^n$ is \textbf{forward invariant} for system~\eqref{sys:auto} if, for initial conditions $\bx_0\in\Cc$, all system trajectories $t\rightarrow\bx(t)$ remain inside the set $\Cc$ for all time $t\geq 0$. The set $\Cc$ is \textbf{safe} if it is also a subset of the safety constraint~$\Sc$.~\hfill$\diamond$
\end{definition}
The idea behind a safe set is that it provides a safe operating region, from which the system can be initialized without violating the safety constraint. In the barrier function framework, we construct the safe set using a continuously differentiable function $\map{h}{\real^n}{\real}$ as:
\begin{equation}\label{eq:safeset-scalar}
    \Cc = \setdefb{\bx\in\real^n}{h(\bx)\geq 0},
\end{equation}
typically with $h(\bx)\geq \psi(\bx),\forall \bx$ to ensure $\Cc\subset \Sc$. Forward invariance, on the other hand, can be established using the barrier condition.
\begin{lemma}\longthmtitle{Barrier Condition}
    \label{lem:BC}
    Consider the autonomous system~\eqref{sys:auto} with a continuous dynamics $\bF$ and the set $\Cc$~in \eqref{eq:safeset-scalar}. If there exists a locally Lipschitz function $\alpha\in\Ke$ such that:
    \begin{equation}\label{eq:BC-scalar}
        L_\bF h(\bx) \geq -\alpha(h(\bx))
    \end{equation}
    for all $\bx$ in an open neighborhood $\Ec\supset\Cc$, then set $\Cc$ is forward invariant for the system.~\hfill$\blacksquare$
\end{lemma}
We adopt the above version of barrier conditions from \cite{PG-JC-ME:17} because it only requires the system dynamics to be continuous, rather than locally Lipschitz.
This is made possible by requiring the barrier condition to hold over an open neighborhood set $\Ec$, instead of $\Cc$, so that the result does not need to rely on Nagumo's theorem. This distinction is particularly relevant for control systems.
\subsection{Control Barrier Function}
We consider the nonlinear control-affine system:
\begin{equation}\label{sys:ctrl-affine}
    \dot \bx = \bf(\bx)+\bg(\bx)\bu
\end{equation}
where $\bx\in\real^n$ is the system state, $\bu\in\real^m$ is the control input, $\map{\bf}{\real^n}{\real^n}$ is the system drift, and $\map{\bg}{\real^n}{\real^{n\times m}}$ is the control matrix.

For control systems, the control input $\bu$ provides additional flexibility for maintaining safety of the system. Analogous to forward invariance for autonomous system, we rely on the concept of control invariance.
\begin{definition}\longthmtitle{Control Invariance}
    A set $\Cc$ is (forward) \textbf{control invariant} for system \eqref{sys:ctrl-affine} if, for any initial condition $\bx(0)\in\Cc$, the system trajectories $t\rightarrow\bx(t)$ can be maintained inside the set $\Cc$ using some corresponding control input $t\rightarrow\bu(t)$. The set $\Cc$ is \textbf{safe} if it is also a subset of the safety constraint $\Sc$.\hfill$\diamond$
\end{definition}

Building on the concept of barrier conditions, control barrier functions are useful for verifying control invariance of a set.
\begin{definition}\longthmtitle{Control Barrier Functions}\label{def:CBF}
    A continuously differentiable function $\map{h}{\real^n}{\real}$ is called a \textbf{control barrier function} (CBF) for system~\eqref{sys:ctrl-affine} if there exists a locally Lipschitz function $\alpha\in\Ke$ such that, for each $\bx$ in the set $\Cc$ defined in~\eqref{eq:safeset-scalar}, there exists a $\bu\in\real^m$ satisfying:
    \begin{equation}\label{eq:CBC}
        \underbrace{L_\bf h(\bx)+L_\bg h(\bx)\bu}_{\dot h(\bx,\bu)} > -\alpha(h(\bx)).
    \end{equation}
\end{definition}
The idea behind CBFs is that they ensure the feasibility of designing a controller to satisfy~\eqref{eq:BC-scalar}. However, the safety result from Lemma~\ref{lem:BC} additionally requires that the controller be continuous and that the barrier condition be satisfied at all points in some neighborhood $\Ec$ outside of $\Cc$. To this end, the modern definition of a CBF uses a strict inequality, as in~\eqref{eq:CBC}, to facilitate the design of controllers that satisfy these condition.

The safety filter framework is one of the main features of CBFs. With CBFs $\{h_i\}_{i=1}^p$, we may construct an optimization-based controller to simultaneously satisfy all constraints:
\begin{align}\label{eq:CBF-QP}
    \bk(\bx)  =\argmin_{\bu\in\real^m} \quad & \|\bu-\bk_\des(\bx)\|^2                                                                  \\
             \textup{s.t.} \quad & \dot h_i(\bx,\bu) \geq -\alpha(h_i(\bx)),~\forall i\in\until{p} \nonumber
\end{align}
where $\map{\bk_\des}{\real^n}{\real^m}$ is a desired continuous controller, without considering safety. The controller~\eqref{eq:CBF-QP} effectively chooses the control input $\bu$ closest to the desired value $\bk_\des(\bx)$ while respecting the safety constraints given by the CBFs.

If the optimization is feasible\footnote{The compatibility between the multiple CBFs is an active research topic and is beyond the scope of this paper.}, the resulting controller is guaranteed continuous. This stems from the strict inequality in~\eqref{eq:CBC} that makes the optimization satisfy Slater's condition at each $\bx$, see \cite[Prop. 2.19]{RAF-PVK:96} and \cite{PM-AA-JC:25}. In addition, the strictness allows for $\bk$ to be well-defined on a neighborhood $\Ec$ of the intersection of safety constraints. As a result, safety filter works synergistically with CBFs to provide a simple construction of a continuous controller that addresses multiple safety constraints simulteneously.

There is also a nonsmooth version of barrier conditions and CBFs where the differentiability requirement on $h$ can be discarded. The details can be found in \cite{PG-JC-ME:17}. For better exposition, we will discuss them when needed later in the paper, rather than in this background section.

\begin{remark}[Lipschitz controllers]
    Local Lipschitz continuity is a desirable property for a controller. Beyond uniqueness of the closed-loop system solutions, Lipschitzness provides useful bounds and is important for system integration, e.g., cascaded or hierarchical system guarantees. The original conception on CBF~\cite{ADA-SC-ME-GN-KS-PT:19} also requires the controller to be locally Lipschitz. Despite the lack of such a guarantee, researchers have employed safety filters with great success. The state of the art results on Lipschitz guarantees are recently provided in~\cite{PM-AA-JC:25}.~\hfill$\bullet$
\end{remark}

\section{Semidefinite Matrix Constraints}
\label{sec:EMCBF}
We introduce matrix-valued barrier functions that enable a richer representation of safety constraints. Rather than using a scalar-valued function, we construct the safe set via a continuously differentiable matrix-valued function $\map{\bH}{\real^n}{\mathbb S^p}$:
\begin{equation}
    \label{eq:safeset-matrix}
    \Cc = \setdefb{\bx\in\real^n}{\bH(\bx)\succeq 0}
\end{equation}
To handle matrix barrier functions, we first introduce the following notation for the entry-wise Lie derivative matrix $\map{\bL_\bF\bH}{\real^n}{\mathbb S^p}$ along a vector field~$\bF$ as:
$$
    [\bL_\bF\bH]_{ij}(\bx) = L_\bF H_{ij}(\bx), \forall i,j\in\until{p}.
$$
Using this notation, we present a barrier condition that maintains matrix semidefiniteness as follows.
\begin{proposition}\longthmtitle{Exponential Semidefinite Matrix Barrier Condition}
    \label{prop:BC-matrix}
    Consider the autonomous system~\eqref{sys:auto} with a continuous vector field $\bF$ and the set $\Cc$ in~\eqref{eq:safeset-matrix}. If the following barrier condition holds with a positive constant $c_\alpha>0$:
    \begin{equation}
        \label{eq:BC-matrix}
        \bL_\bF\bH(\bx) \succeq -c_\alpha\bH(\bx)
    \end{equation}
    for all $\bx$ in an open neighborhood $\Ec\supset\Cc$, then set $\Cc$ is forward invariant for the system.
\end{proposition}
\begin{proof}
    For any given vector~$\bv\in\real^p$, we define an auxiliary function:
    \begin{equation}\label{eq:aux}
        \xi_\bv(\bx) = \bv^\top \bH(\bx)\bv
    \end{equation}
    Letting $t\rightarrow \bx(t)$ be any solution from an initial condition $\bx_0\in\Cc$, we have:
    \begin{align*}
        \frac{d}{dt}(\xi_\bv \circ \bx)(t) & = \bv^\top \frac{d}{dt} (\bH \circ \bx)(t) \bv \\
                                           & = \bv^\top \bL_\bF\bH(\bx(t)) \bv              \\
                                           & \geq -c_\alpha\bv^\top (\bH \circ \bx)(t) \bv  \\
                                           & = -c_\alpha (\xi_\bv \circ \bx)(t)
    \end{align*}
    where we have used the matrix inequality~\eqref{eq:BC-matrix} to derive the inequality for all time $t$ such that $\bx(t)\in\Ec$.

    Suppose there exists a solution $\bx(t)$ starting from $\bx(0)=\bx_0\in\Cc$ yet  $\bx(t^*)\not \in \Cc$ at some time $t^*\geq 0$. Then there exists $\bv^*$ such that $\xi_{\bv^*}(\bx(t^*))<0$. By continuity of the solution, there exists a time interval $[t_{\partial\Cc},t_\Ec] \subseteq [t_0,t^*]$ such that $\xi_{\bv^*}(\bx(t_{\partial\Cc}))=0$ and $\xi_{\bv^*}(\bx(t))<0$ for all $t\in(t_{\partial\Cc},t_\Ec]$ which, by the definition of th positive definiteness, indicates $\bx(t)\in\Ec\setminus \Cc, \forall t\in(t_{\partial\Cc},t_\Ec]$. However, the inequality above suggests $\xi_{\bv^*}\circ \bx$ must be increasing during time $t\in(t_{\partial\Cc},t_\Ec)$, which is a contradiction since $\xi_{\bv^*}(\bx(t_\Ec))<0=\xi_{\bv^*}(\bx(t_{\partial\Cc}))$.
\end{proof}
\cref{prop:BC-matrix} establishes a matrix  barrier condition~\eqref{eq:BC-matrix} for the semidefinite matrix safety constraint~\eqref{eq:safeset-matrix}. Analogous to the scalar case, we propose the following definition for matrix control barrier functions.
\begin{definition}\longthmtitle{Exponential MCBF}
    A continuously differentiable function $\map{\bH}{\real^n}{\mathbb S^p}$ is called an \textbf{exponential matrix control barrier function} (exponential MCBF) for system~\eqref{sys:ctrl-affine} if there exists a constant $c_\alpha> 0$ such that, for each $\bx$ in the set $\Cc$ defined in~\eqref{eq:safeset-matrix}, there exists a $\bu\in\real^m$ satisfying:
    \begin{equation}\label{eq:CBC-matrix}
        \underbrace{\bL_\bf\bH(\bx) +\sum_{i=1}^m\bL_{\bg_i}\bH(\bx)u_i}_{\dot \bH(\bx,\bu)} \succ -c_\alpha\bH(\bx).
    \end{equation}
    where $\bg_i$ is the $i$-th column of the control matrix $\bg(\bx)$.~\hfill$\diamond$
\end{definition}
The definition of MCBFs with a strictly positive definite condition helps ensure the existence of a continuous controller, establishing safety of the set $\Cc$.

\begin{theorem}\longthmtitle{Safety from exponential MCBF}
    \label{thm:EMCBF-safety}
    Consider the control-affine system~\eqref{sys:ctrl-affine} and the set  $\Cc$ in~\eqref{eq:safeset-matrix}. If $\bH$ is an exponential MCBF, then set $\Cc$ is control invariant. In particular, the CBF-based semidefinite programming (CBF-SDP) safety filter:
    \begin{align}\label{eq:CBF-SDP}
        \bk(\bx) = \argmin_{\bu\in\real^m} \quad &\|\bu-\bk_\des(\bx)\|^2\hfill\nonumber      \\
                 \textup{s.t.} \quad & \dot \bH(\bx,\bu) \succeq -c_\alpha\bH(\bx)
    \end{align}
    is continuous for all $\bx$ in some neighborhood $\Ec$ of $\Cc$. Consequently, the state-feedback $\bu=\bk(\bx)$ renders set $\Cc$ forward invariant for the closed-loop system.
\end{theorem}
\begin{proof}
    Define the set-valued map:
    \begin{align*}
        \Uc(\bx) & = \setdefb{\bu\in\real^m}{\dot\bH(\bx,\bu) +c_\alpha\bH(\bx) \succeq \bzero}        \\
                 & =\setdefb{\bu\in\real^m}{\phi_1\big(\dot\bH(\bx,\bu) +c_\alpha\bH(\bx)\big)\geq 0}.
    \end{align*}
    where $\phi_1$ is the function returning the minimum eigenvalue of a given matrix. Since the matrices are entry-wise continuously differentiable and the eigenvalue function~$\phi_1$ is globally Lipschitz (see e.g., \cite[Thm. 1]{MDS-JC:09} and \cite[Thm. 2.4]{ASL:96}), the resulting function $\phi_1\circ(\dot \bH+c_\alpha\bH)$ is continuous. As such, not only does the map have a nonempty interior for $\bx\in\Cc$ from the definition of MCBF, but it also has a nonempty interior on some neighborhood $\Ec\supset \Cc$ from continuity. Hence, we conclude the map $\Uc$ is lower semicontinuous on $\Ec$ because it has a nonempty interior, is convex-valued, and the function $\phi_1\circ(\dot \bH+c_\alpha\bH)$ defining it is continuous, see \cite[Lem. 5.2]{GS:18}. After substituting $\Delta \bu = \bu-\bk_\des(\bx)$ as a new variable, we have a minimal selection formulation:
        \begin{align*}
            \Delta\bk(\bx) = \argmin_{\bu\in\real^m} \quad &\|\Delta\bu\|^2\hfill\nonumber                                          \\
                          \textup{s.t.} \quad & \phi_1(\dot \bH\big(\bx,\Delta\bu+\bk_\des(\bx))+c_\alpha\bH(\bx)\big) \geq 0.
        \end{align*}
    Here, the constraint remains lower semicontinuous, convex-valued, and closed-valued.  Therefore, from \cite[Prop. 2.19]{RAF-PVK:96}, its minimal selection on a Euclidean space is continuous for all $\bx\in\Ec$, implying $\bk(\bx)=\Delta\bk(\bx)+\bk_\des(\bx)$ is also continuous, as desired.

    Forward invariance of the closed-loop system $\bF(\bx) = \bf(\bx)+\bg(\bx)\bk(\bx)$ follows from Lemma~\ref{prop:BC-matrix}, and control invariance follows from the existence of the control signal $t\rightarrow\bu(t)=\bk(\bx(t))$, concluding the proof.
\end{proof}
Theorem~\ref{thm:EMCBF-safety} formally states the safety result established by an exponential MCBF, proposing specifically the CBF-SDP safety filter~\eqref{eq:CBF-SDP} as one of the possible controllers. Beyond safety guarantees, CBFs typically offer useful bounds on the evolution of safety. To complement our results, we next provide  the bounds associated with MCBF-based controllers.

\subsection{Spectral Analysis}
For the standard CBF condition~\eqref{eq:BC-scalar}, we can derive an exponential bound on the evolution of $h$ along the trajectory. This occurs when $\alpha(h(\bx))= c_\alpha h(\bx)$ is chosen as a linear function with $c_\alpha >0$. That is, we have $h(\bx(t))\geq h(\bx_0)e^{-c_\alpha t}$ from an initial condition $\bx_0\in\real^n$. Incidentally, we also refer to our MCBF in~\eqref{eq:CBC-matrix} as exponential, despite $\bH$ not being scalar-valued. This is motivated by how the evolution of its eigenvalues can be similarly exponentially bounded.

For each $\bx\in\real^n$, $\bH(\bx)$ is a symmetric matrix and thus has real eigenvalues. Formally, we use $\map{\phi_j}{\mathbb S^p}{\real}$ to denote the function mapping a symmetric matrix $\bH$ to its $j$-th smallest eigenvalue. At the same time, we use $\map{\lambda_j}{\real^n}{\real}$ for the state-dependent eigenvalue function through a function composition: $\lambda_j(\bx)=\phi_j(\bH(\bx))$. Without loss of generality, we define $\{\lambda_j\}_{j=1}^p$ so that they are in an ascending order:
\begin{equation}
    \lambda_1(\bx)\leq \dots \leq \lambda_p(\bx),~\forall \bx\in\real^n.
\end{equation}
We note importantly that the function $\{\lambda_j\}_{j=1}^p$ is not differentiable at points where the eigenvalue  is not simple. Therefore, we rely on nonsmooth analysis to bound its rate of change along an autonomous system~\eqref{sys:auto} as:
$$
    \frac{d}{dt}(\lambda_j\circ \bx)(t) \in \partial \phi_j (\bH(\bx(t))) \cdot \bL_\bF\bH(\bx(t)),  ~a.e.~t\geq0,
$$
with its weak set-valued Lie derivative~\cite[Lem. 1 and Rmk. 1]{PG-JC-ME:17}, constructed by the generalized gradient set $\partial\phi_j$. Note that the left hand side is also guaranteed to be absolutely continuous. For the right hand side, we have used the nonsmooth chain rule~\cite[Thm. 2.3.10]{FHC:83} for the weak set-valued Lie derivative since we know the generalized gradient of an eigenvalue of a symmetric matrix function, cf.~\cite{MDS-JC:09}:
$$
    \partial\phi_j(\bH) = \text{co}\setdefb{\bv\bv^\top \in \mathbb{S}^p}{\bH\bv=\phi_j(\bH)\bv,~\|\bv\| = 1}.
$$
Since the vectors $\bv$'s defining in the set above are essentially the normalized eigenvectors associated with $\phi_j$, we have the following result. 

\begin{proposition}\longthmtitle{Exponential Bound on Eigenvalues}\label{prop:exp_bound}
    Consider the autonomous system~\eqref{sys:auto} with continuous dynamics $\bF$ and the set $\Cc$ in~\eqref{eq:safeset-matrix}. If the matrix barrier condition~\eqref{eq:BC-matrix} holds, then the eigenvalues can be bounded as:
    \begin{equation}\label{eq:exp_bound}
        \lambda_i(\bx(t)) \geq \lambda_i(\bx_0)e^{-c_\alpha t},~\forall i\in\until p,
    \end{equation}
    at almost every time $t\geq0$,  along any Carath\'eodory solution starting from $\bx_0\in\Cc$.
\end{proposition}
\begin{proof}
    Let $\bA(t)\in\partial\phi_j(\bH(\bx(t))$ be such that
    $$
        \frac{d}{dt}(\lambda_j\circ \bx)(t) = \bA(t)\cdot \dot \bH(\bx(t),\bk(\bx(t)))
    $$
    at almost every time $t\geq 0$. Since $\bA$ is an element in the convex hull, there may not be an eigenvector $\bv$ such that $\bA=\bv\bv^\top$. Therefore, we use the Carath\'eodory theorem of convex hulls~\cite[Thm. 17.1]{RTR:70} to deduce the existence of  vectors $\{\bv_s(t)\}_{s=1}^{(p^2+p)/2}$ such that
    $$
        \bA(t) = \sum_{s=1}^{(p^2+p)/2} a_s(t) \bv_s(t)\bv_s(t)^\top
    $$
    where $\sum a_s(t) = 1$ and each $\bv_s$ is a normalized eigenvector associated with $\phi_j(\bH(\bx))$. Using this fact, we bound:
    \begin{align*}
        \bA\cdot \bL_\bF\bH(\bx) & = \Big(\sum_{s=1}^{(p^2+p)/2} a_s \bv_s\bv_s^\top\Big) \cdot \bL_\bF\bH(\bx) \\
                                 & = \sum_{s=1}^{(p^2+p)/2} a_s \big(\bv_s\bv_s^\top \cdot \bL_\bF\bH(\bx)\big) \\
                                 & = \sum_{s=1}^{(p^2+p)/2} a_s \big(\bv_s^\top \bL_\bF\bH(\bx)\bv_s\big)       \\
                                 & \geq  -c_\alpha\sum_{s=1}^{(p^2+p)/2} a_s \big(\bv_s^\top \bH(\bx)\bv_s\big) \\
                                 & = -c_\alpha\sum_{s=1}^{(p^2+p)/2} a_s \lambda_j(\bx)                   = -c_\alpha \lambda_j(\bx)
    \end{align*}
    where we have dropped the dependencies on $t$ for compactness of the presentation. Hence, the comparison lemma ensures the evolution bound~\eqref{eq:exp_bound}, concluding the proof.
\end{proof}
Proposition~\ref{prop:exp_bound} shows that the proposed matrix barrier condition~\eqref{eq:BC-matrix} enforces bounds on all eigenvalues of $\bH$. The bounds suggest a degree of conservatism in the approach, since ensuring only the smallest eigenvalue $\lambda_1$ nonnegative is sufficient (and necessary) to keep $\bH$ positive semidefinite. Nevertheless, it is typically difficult to handle the nonsmooth nature associated with eigenvalues, and our approach offers a simple way to address it. We further expand on this point by focusing our discussion on diagonal matrices.

\subsection{Diagonal Matrix Constraints}
Diagonal matrices are special cases of  symmetric matrices~$\bH$. When $\bH$ is diagonal, the safe set $\Cc$ in~\eqref{eq:safeset-matrix} can be equivalently represented as:
$$
    \Cc = \setdefb{\bx\in\real^n}{\bH_{ii}(\bx) \geq 0, \forall i\in\until{p}}.
$$
Here, each $\bH_{ii}$ can be viewed as an individual scalar-valued CBF because the matrix barrier condition~\eqref{eq:CBC-matrix} is equivalent to
$$
    \dot \bH_{ii}(\bx,\bu) \geq -c_\alpha\bH_{ii}(\bx),~\forall i\in\until p,
$$
which is precisely what we would have when considering multiple exponential CBFs with linear $\alpha$  simultaneously. The following corollary makes this observation formal.
\begin{corollary}\longthmtitle{Multiple CBFs via MCBF}\label{cor:diag}
    Consider the control-affine system~\eqref{sys:ctrl-affine}. Given a set of control barrier functions $\{h_i\}_{i=1}^p$ with $\alpha_i(r)=c_\alpha r$ for all $i\in\until p$, let a diagonal matrix $\bH$ be constructed with $\bH_{ii}(\bx)=h_i(\bx)$. Then $\bH$ is an exponential MCBF, and the safety filters \eqref{eq:CBF-QP} and \eqref{eq:CBF-SDP} are equivalent and continuous on $\Cc= \bigcap_{i=1}^p \Cc_i$ defined in~\eqref{eq:safeset-matrix}. As a result, the state-feedback $\bu=\bk(\bx)$ renders set $\Cc$  forward invariant for the closed-loop system.~\hfill$\blacksquare$
\end{corollary}
Corollary~\ref{cor:diag} formalizes the common strategy of using multiple CBFs to simultaneously handle multiple safety constraints. While it is generally understood that all barrier conditions should all hold together, a frequently overlooked detail is that they only need to do so on a neighborhood of the intersection~$\bigcap_{i=1}^p \Cc_i$. It is unnecessary to verify that each CBF hold over its entire individual safe set~$\Cc_i$. Although this nuance is familiar to many practitioners, it is rarely stated formally. Our diagonal MCBF formulation makes this composition process explicit and provides a formal justification for this widely used strategy.

Another related strategy for simultaneously addressing multiple safety constraints is through Boolean-based CBFs~\cite{PG-JC-ME:17}, which also connects to our diagonal matrix formulation. In particular, the smallest eigenvalue of a diagonal matrix $\bH$ is:
$$
    \lambda_1(\bx) = \min_{i\in\until{p}} \bH_{ii}(\bx) \defeq h_{\min}(\bx).
$$
Instead of considering all CBFs at the same time, we may choose to only focus on the one with the smallest value. This particular approach corresponds to the conjunctive (AND) Boolean composition discussed in ~\cite{PG-JC-ME:17}. To address the nonsmoothness introduced when the index of the smallest CBF switches, \cite{PG-JC-ME:17} proposes a barrier condition that essentially considers the worst-case $\frac{d}{dt}h_{\min}(\bx(t))$. However, such formulation does not yield continuity in the optimization-based controller, see \cite{PO-BC-LS-JC:23-auto} for an example.

Our MCBF formulation offers one approach to design continuous controllers for Boolean-based CBFs~\cite{PG-JC-ME:17} without relaxing the safe set using soft-min or soft-max approximation like in~\cite{TGM-ADA:23}. Undeniably, as suggested in Corollary~\ref{cor:diag}, our proposed formulation recovers the safety filter~\eqref{eq:CBF-QP} we typically use for dealing with multiple CBFs. However, as we shall discuss in later section, the novelty of our approach manifests when we consider disjunctive (OR) Boolean compositions.

\subsection{Lipschitz Controllers}
This paper does not establish Lipschitz guarantees for the CBF-SDP~\eqref{eq:CBF-SDP}. Regularity of an optimization-based controller is an active field of research. For instance, it was only recently rediscovered that a CBF-QP with multiple CBF inequalities~\eqref{eq:CBF-QP} is locally Lipschitz under the linear independence constraint qualification (LICQ) assumption~\cite{PM-AA-JC:25}. At the same time, there is little work on the regularity of more complex optimization-based controller, such as second-order cone programming (SOCP) formulations for measurement-robust CBF~\cite{RKC-AWS-AJT-TGM-KLB-ADA:21}. Given the inherent difficulty of identifying conditions that guarantee Lipschitz continuity---even for the diagonal case---we decide to leave it as an open problem and a part of our future investigations on MCBFs. Instead, we provide here a result establishing the existence of a smooth (though not necessarily optimization-based) controller associated with a MCBF.

\begin{proposition}
    \label{prop:artstein}
    Consider the control-affine system~\eqref{sys:ctrl-affine} and the set $\Cc$ in~\eqref{eq:safeset-matrix}. If $\bH$ is an exponential MCBF, then there exists a smooth controller $\map{\bk}{\real^n}{\real^m}$ such that $\bu=\bk(\bx)$ satisfies~\eqref{eq:CBC-matrix} for all $\bx$ in an open neighborhood $\Ec\supset \Cc$. Consequently, set $\Cc$ is control invariant for the system.
\end{proposition}
\begin{proof}
    From the definition of an exponential CBF, there exists a function $\bx\rightarrow\bu^*(\bx)$, not necessarily continuous, such that $\dot\bH(\bx,\bu^*(\bx)) \succ -c_\alpha\bH(\bx)$. Then because the function $\dot\bH(\bx,\bu)+ c_\alpha\bH(\bx)$ is continuous in $\bx$, there exists a neighborhood $\Wc(\bx)$, for each $\bx$, such that $\dot\bH(\bx',\bu^*(\bx))+ c_\alpha\bH(\bx')$ remains positive definite for all $\bx'\in\Wc(\bx)$. Since the $\bu^*(\bx)$ is a feasible control in its corresponding neighborhood $\Wc(\bx)$, we can invoke the arguments of Artstein's theorem~\cite[Thm. 4.1]{ZA:83} (see also~\cite[Lem. 6.5]{PO-BC-LS-JC:23-auto}) to deduce the existence of a smooth controller through a finite partition of unity, relying on the fact that a LMI is a convex constraint.
\end{proof}
Proposition~\ref{prop:artstein} relies on the Artstein's theorem to non-constructively establish the existence of a smooth controller. This result shows that smooth controllers are theoretically achievable. While challenging, identifying conditions guaranteeing Lipschitz continuity for optimization-based controllers is a valuable and promising area for future research.

\section{Indefinite Matrix Constraints}
Often times, safety requirements involve maintaining state trajectories outside a given set, e.g., obstacle avoidance, rather than inside the set. For scalar-valued CBFs, there is no distinction between the two scenarios because CBFs can be trivially reformulated with a negation. However, a negation to a matrix requires extra care.

\subsection{Disjunctive Boolean on CBFs}
We begin with a motivating example. Consider a cylinder described by the following three inequalities:
\begin{align*}
    x_1^2+x_2^2 \leq 1 ~\text{and}~ -1 \leq x_3\leq 1.
\end{align*}
To keep state trajectories within this cylinder, we can rely on the positive semidefinite constraint formulation developed in this paper. On the other hand, when this cylinder represents an obstacle, we would like the state trajectories to remain outside of it. This leads to a safety constraint given by the logical negation of the cylinder representation:
\begin{align*}
    x_1^2+x_2^2 \geq 1 ~\text{or}~ x_3\leq -1~\text{or}~ 1\leq x_3,
\end{align*}
introducing OR Booleans among the inequalities. To tackle this problem using the standard CBF approach, the common practice is to smoothly approximate the safe set with a soft-max function on all three constraints, see~\cite{TGM-ADA:23}.

From the new perspective using a matrix representation, we may also represent the negation of a (open) cylinder with an indefinite inequality:
$$
    -\begin{bmatrix}
        1-x_1 & x_2   & 0     & 0     \\
        x_2   & 1+x_1 & 0     & 0     \\
        0     & 0     & 1-x_3 & 0     \\
        0     & 0     & 0     & 1+x_3
    \end{bmatrix} \not \prec~ \bzero.
$$ Motivated by this class of problems, we develop a MCBF framework to deal with such constraints.
\subsection{Indefinite CBFs}
In this section, we consider safe sets given by:
\begin{align}\label{eq:safeset-indefinite}
    \Cc & = \setdefb{\bx\in\real^n}{\bH(\bx)\not\prec \bzero} \nonumber \\
        & = \setdefb{\bx\in\real^n}{\lambda_p(\bx)\geq 0}
\end{align}
The equivalent eigenvalue formulation facilitates our subsequent result for indefinite matrices.
\begin{proposition}\longthmtitle{Indefinite Matrix Barrier Condition}
    \label{prop:BC-indefinite}
    Consider the autonomous system~\eqref{sys:auto} with a continuous vector field $\bF$  and the set $\Cc$ in~\eqref{eq:safeset-indefinite}. If there exists a function $\alpha\in\Ke$ and $c_\perp \geq 0$ such that:
    \begin{equation}
        \label{eq:BC-indefinite}
        \bL_\bF\bH(\bx) \succeq -\alpha\big(
        \lambda_p(\bx)
        \big)\bI_{p\times p} -c_\perp\big(
        \lambda_p(\bx)\bI_{p\times p}-\bH(\bx)
        \big),
    \end{equation}
    for all $\bx$ in an open neighborhood $\Ec\supset\Cc$, then the following bound holds:
    \begin{equation}\label{eq:indefinite-bound-largest}
        \frac{d}{dt}(\lambda_p\circ \bx)(t) \geq -\alpha\big(
        (\lambda_p\circ\bx)(t)
        \big)
    \end{equation}
    at almost every time $t\geq0$ along any Carath\'eodory solution starting from $\bx_0\in\Cc$. As a consequence, set $\Cc$ is forward invariant for the system.
\end{proposition}
\begin{proof}
    Following the proof of Proposition~\ref{prop:exp_bound}, we may bound the eigenvalues $\{\lambda_j\}_{j=1}^p$ along the trajectory at almost every time $t\geq 0$ as:
    \begin{equation}~\label{eq:indefinite-bound}
        \frac{d}{dt}(\lambda_j\circ \bx)(t) \geq -\alpha\big(
        (\lambda_p\circ\bx)(t)
        \big)-c_\perp\big(
        \lambda_p(\bx)-\lambda_j(\bx)
        \big).
    \end{equation}
    In particular, we derive~\eqref{eq:indefinite-bound-largest} for the largest eigenvalue with $j=p$.
    Hence, from the comparison lemma with any initial condition~$\bx_0\in\Cc$, the eigenvalue $\lambda_p$ must remain nonnegative at all time, preventing $\bH$ from becoming negative definite.
\end{proof}
The key idea behind the matrix inequality in~\eqref{eq:BC-indefinite} is to isolate the largest eigenvalue and prevent it from becoming negative. This can, in fact, be established even without the last term in the inequality. The last term, however, helps relax the barrier condition further. Without it, the inequality would also unnecessarily bound the evolution of all other eigenvalues, as evident in~\eqref{eq:indefinite-bound}. In a graceful manner, the relaxation vanishes as $\lambda_j$ approaches $\lambda_p$.

Note importantly that the barrier condition avoids nonsmoothness of $\lambda_p$ through relying on $\bL_\bF\bH$ rather than $\frac{d}{dt}(\lambda_p\circ \bx)$ directly. This formulation lays the foundation for the MCBF framework that facilitates the synthesis of a continuous controller.

\begin{definition}\longthmtitle{Indefinite MCBF}
    A continuously differentiable function $\map{\bH}{\real^n}{\mathbb S^p}$ is called an \textbf{indefinite matrix control barrier function} (indefinite MCBF) for system~\eqref{sys:ctrl-affine} if there exists a function $\alpha\in\Ke$ and $c_\perp \geq 0$ such that, for each $\bx$ in the set $\Cc$ defined in~\eqref{eq:safeset-indefinite}, there exists a $\bu\in\real^m$ satisfying:
    \begin{equation}
        \dot\bH(\bx,\bu) \succ -\alpha\big(
        \lambda_p(\bx)
        \big) \bI_{p\times p} -c_\perp\big(
        \lambda_p(\bx)\bI_{p\times p}-\bH(\bx)
        \big).
    \end{equation}
\end{definition}
Much like the semidefinite case, indefinite MCBFs facilitate the construction of a continuous controller.
\begin{theorem}\longthmtitle{Safety from Indefinite MCBF}\label{thm:safety_indefinite-MCBF}
    Consider the control-affine system~\eqref{sys:ctrl-affine} and the set $\Cc$ in~\eqref{eq:safeset-indefinite}.
    If $\bH$ is an indefinite MCBF, then set $\Cc$ is control invariant. In particular, the CBF-SDP:
    \begin{align}
        \label{eq:CBF-SDP-indefinite}
        \bk(\bx)  = \argmin_{\bu\in\real^m} \quad & \|\bu-\bk_\des(\bx)\|^2            \\
                 \textup{s.t.} \quad & \dot \bH(\bx,\bu) \succeq -\alpha(
        \lambda_p(\bx))\bI_{p\times p}\nonumber                                \\
                 & \qquad\qquad\qquad -c_\perp\left(
        \lambda_p(\bx)\bI_{p\times p}-\bH(\bx)
        \right)\nonumber
    \end{align}
    is continuous for all $\bx$ in some neighborhood $\Ec$ of $\Cc$. Consequently, the state-feedback $\bu=\bk(\bx)$ renders set $\Cc$ forward invariant for the closed-loop system.~\hfill$\blacksquare$
\end{theorem}
Theorem~\ref{thm:safety_indefinite-MCBF} provides a continuous controller for dealing with indefinite matrix safety constraints. We omit the proof for the theorem because of its similarities to the one for Theorem~\ref{thm:EMCBF-safety}. An important application of indefinite MCBFs is on a diagonal matrix constraint constructed by multiple CBFs.
\begin{corollary}\longthmtitle{OR Boolean CBFs via MCBF}\label{cor:bool-or}
    Consider the control-affine system~\eqref{sys:ctrl-affine}. Given a set of control barrier functions $\{h_i\}_{i=1}^p$ with some corresponding $\alpha_i\in\Ke$ for each $i\in\until p$, let a diagonal matrix $\bH$ be constructed with $\bH_{ii}(\bx)=-h_i(\bx)$. If $\bH$ is an indefinite MCBF with $\alpha(r)\defeq\max_{i\in\until{p}}\alpha_i(r)$ and some constant $c_\perp\geq 0$ on the  set $\Cc= \cup_{i=1}^p \Cc_i$ defined in~\eqref{eq:safeset-matrix}. Then, the state-feedback $\bu=\bk(\bx)$ using the safety filter \eqref{eq:CBF-SDP-indefinite} renders the  set $\Cc$ forward invariant for the closed-loop system.~\hfill$\blacksquare$
\end{corollary}
Our MCBF formulation enables the synthesis of continuous controllers that handle nonsmooth Boolean-based constraints. In particular, Corollary~\ref{cor:bool-or} offers a mechanism for addressing OR  Boolean constraints where it suffices to maintain the positivity of any one CBF at a given time. In addition, the approach provides a bound on the evolution of the maximum eigenvalue $\lambda_p$ as in~\eqref{eq:indefinite-bound-largest}. This coincides with the CBF that has a largest value at a given state, i.e., $h_{\max}(\bx) \defeq \max_{i\in \until{p}} h_i(\bx)$. Notice that the bound~\eqref{eq:indefinite-bound-largest} obtained involves the more general class-$\Ke$ function, rather than a linear one, because we deal directly with eigenvalues. This insight motivates the developments in the next section.

\section{General Matrix Control Barrier Functions}
A limitation of exponential MCBFs developed earlier is that their construction~\eqref{eq:CBC-matrix} only permits a linear relationship between $\dot \bH$ and $\bH$.
This section formulates the general class of MCBFs that allows the use of class-$\Ke$ functions.

\subsection{Smallest Eigenvalue MCBF}
Building on the development from the last section, we propose a barrier condition that maintains positivity of the lowest eigenvalue.
\begin{proposition}\longthmtitle{Smallest Eigenvalue Matrix Barrier Condition}
    \label{prop:BC-smallest}
    Consider the autonomous system~\eqref{sys:auto} with a continuous vector field $\bF$  and the set $\Cc$ in~\eqref{eq:safeset-matrix}. If there exists a function $\alpha\in\Ke$ and $c_\perp \geq 0$ such that:
    \begin{equation}\label{eq:BC-smallest}
        \bL_\bF\bH(x) \succeq -\alpha\big(
        \lambda_1(\bx)
        \big)\bI_{p\times p}-c_\perp\big(
        \bH(\bx)-\lambda_1(\bx)\bI_{p\times p}
        \big)
    \end{equation}
    for all $\bx$ in an open neighborhood $\Ec\supset\Cc$, then the following bound holds:
    \begin{equation}\label{eq:eig-bound-smallest}
        \frac{d}{dt}(\lambda_1\circ \bx)(t) \geq -\alpha\big(
        (\lambda_1\circ\bx)(t)
        \big).
    \end{equation}
    at almost every time $t\geq0$ along any Carath\'eodory solution starting from $\bx_0\in\Cc$. As a consequence, set $\Cc$ is forward invariant for the system.~\hfill$\blacksquare$
\end{proposition}
While the barrier condition proposed in Proposition~\ref{prop:BC-smallest} successfully introduces a class-$\Ke$ function to the bound, it only provides a bound on the smallest eigenvalue. With a slight adjustment of setting $c_\perp=0$, bounds can be developed for other eigenvalues since $\alpha(\lambda_1(\bx))\leq \alpha(\lambda_j(\bx))$ for any $j\in\until{p}$. However, such an approach is inherently conservative as it essentially bounds the evolution of all eigenvalues with the smallest eigenvalue $\lambda_1$ rather than directly with each $\lambda_j$. To this end, we omit the result on the corresponding MCBF associated with this barrier condition.

\subsection{General Matrix CBFs}
The more appropriate approach to bounding each eigenvalue involves matrix diagonalization through spectral decomposition. For a given symmetric matrix $\bH(\bx)$, we have:
$$
    \bH(\bx) = \bV(\bx)\bLambda(\bx)\bV(\bx)^\top
$$
where $\bLambda(\bx)$ is a diagonal matrix constructed from eigenvalues, $\Lambda_{ii}(\bx)=\lambda_i(\bx)$, and $\bV(\bx)$ is a matrix constructed from concatenating corresponding eigenvectors. Once diagonalized, we may apply the class-$\Ke$ function on each eigenvalue. We denote with $\bLambda_\alpha(\bx)$ the resulting matrix from applying the function $\alpha\in\Ke$ element-wise on the matrix $\bLambda(\bx)$:
    $$
        \bLambda_\alpha(\bx) \defeq \diag\big(\alpha(\lambda_1(\bx)),\dots,\alpha(\lambda_p(\bx))\big).
    $$
    Then we use the following notation for applying a class-$\Ke$ function to a matrix:
\begin{equation}
\label{eq:balpha}
        \balpha(\bH(\bx)) \defeq \bV(\bx)\bLambda_\alpha(\bx)\bV(\bx)^\top,
\end{equation}
    for which we can develop a barrier condition that mirrors the scalar case.
\begin{proposition}\longthmtitle{Matrix Barrier Condition}
    \label{prop:BC-diagonalization}
    Consider the autonomous system~\eqref{sys:auto} with a continuous vector field $\bF$  and the set $\Cc$ in~\eqref{eq:safeset-matrix}. If there exists a function $\alpha\in\Ke$ such that:
    \begin{equation}\label{eq:BC-diagonalization}
        \bL_\bF\bH(x) \succeq -\balpha(\bH(\bx))
    \end{equation}
    for all $\bx$ in an open neighborhood $\Ec\supset\Cc$, then the following bound holds:
    \begin{equation}\label{eq:eig-bound-all}
        \frac{d}{dt}(\lambda_j\circ \bx)(t) \geq -\alpha((\lambda_j\circ\bx)(t)), \forall j\in\until{p},
    \end{equation}
    at almost every time $t\geq0$ along any Carath\'eodory solution starting from $\bx_0\in\Cc$. As a consequence, set $\Cc$ is forward invariant for the system.
\end{proposition}
\begin{proof}
    Since the matrix $\bV(\bx)$ is orthonormal at each $\bx$, we can deduce the evolution bound \eqref{eq:eig-bound-all} using the logic of the proof for Proposition~\ref{prop:exp_bound}.
\end{proof}
Proposition~\ref{prop:BC-diagonalization} lays the necessary groundwork for the development of a general class of MCBFs that can impose a class-$\Ke$ function bound on all of the eigenvalues. Note importantly that the class-$\Ke$ function must be uniform for all eigenvalues, unlike in the diagonal case.
\begin{remark}\longthmtitle{Uniform class-$\Ke$ function}
    Diagonal matrices enjoy the luxury of selecting different $\alpha_i$ for each eigenvalue $\lambda_i$. For these matrices, it is possible to correspond the eigenvalues with eigenvectors. However, the correspondence is unclear for general symmetric matrices, especially when there are repeated eigenvalues. By imposing different $\alpha_i$, repeated eigenvalues scale to different values. The uniqueness of the resulting matrix $\balpha(\bH(\bx))$ requires careful labeling of eigenvalues and associated eigenvectors. In some cases, this is possible, but we choose to avoid this for simplicity.~\hfill$\bullet$
\end{remark}

\begin{definition}\longthmtitle{Matrix CBF}
    A continuously differentiable function $\map{\bH}{\real^n}{\mathbb S^p}$ is called a \textbf{matrix control barrier function} (MCBF) for system~\eqref{sys:ctrl-affine} if there exists a function $\alpha\in\Ke$ such that, for each $\bx$ in the set $\Cc$ defined in~\eqref{eq:safeset-matrix}, there exists a $\bu\in\real^m$ satisfying:
    \begin{equation}\label{eq:MCBF}
        \dot\bH(\bx,\bu) \succ -\balpha(\bH(\bx)).
    \end{equation}
    with $\balpha(\bH(\bx))$ defined in \eqref{eq:balpha}.
\end{definition}
The introduction of the eigenvector matrix $\bV$ raises an important concern regarding the continuity of the $\balpha\circ \bH$. Such a guarantee facilitates establishing the lower semicontinuity of the constraint set generated by the MCBF, thereby ensuring continuity of the resulting safety filter.
The following result resolves this issue.

\begin{lemma}\longthmtitle{Continuity of the Matrix Class-$\Ke$ Functions}\label{lem:continuity_class-K}
    Let the matrix-valued function $\bH$ be continuous and $\alpha$ be a class-$\Ke$ function. The matrix-valued function $\balpha\circ\bH$ is continuous if $\{\bv_i\}_{j=1}^p$ are selected orthonormal.
\end{lemma}
\begin{proof}
    First, we write the matrix function as:
    \begin{align*}
        \balpha(\bH(\bx)) & = \sum_{j=1}^p \alpha(\lambda_j(\bx))\bv_j(\bx)\bv_j(\bx)^\top \\
                                & = \sum_{j=1}^p \alpha(\lambda_j(\bx))\bP_j(\bx)
    \end{align*}
    in terms of functions $\bP_j(\bx)\defeq\bv_j(\bx)\bv_j(\bx)^\top$. The matrix $\bP_j$ is precisely the eigenspace projector when $\lambda_j$ has a multiplicity of one. For higher multiplicity, $\bP_j$ is not unique, so depending on the selection of $\bv_j$, it may not be continuous. Because $\bv_j$ are orthonormal and span the entire eigenspace, the eigenspace projector is then given by the sum of the projectors with the same eigenvalue, i.e., $\sum_{j\in\mathcal{J(\bx)}}\bP_j(\bx)$ where $\mathcal{J}(\bx) = \setdef{i\in\until{p}}{ \lambda_i(\bx)=\lambda_j(\bx)}$. From \cite[Chapter 2 Sec. 1]{TK:13} (see the subsection 8 for a brief summary), the total projector $\sum_{j\in\mathcal{J(\bx)}}\bP_j(\bx)$ is holomorphic at $\bx$ and adopts smoothness of $\bH$ on the manifold where total eigenvalue multiplicity $|\mathcal{J}(\bx)|$ is constant. Therefore, since $\alpha(\lambda_i(\bx))=\alpha(\lambda_j(\bx))$ for all $i\in\mathcal J(\bx)$, we may group $\bP_j(\bx)$'s with the summation $\sum_{j\in\mathcal J(\bx)}\bP_j(\bx)$ to deduce continuity of $\balpha\circ \bH$ as desired.
\end{proof}
With this continuity result, the safety result follows with a similar proof to earlier theorems on safety.
\begin{theorem}\longthmtitle{Safety from MCBF}\label{thm:safety_MCBF}
    Consider the control-affine system~\eqref{sys:ctrl-affine} and the set $\Cc$ in~\eqref{eq:safeset-matrix}.
    If $\bH$ is an MCBF, then set $\Cc$ is control invariant. In particular, the CBF-SDP:
    \begin{align}
        \label{eq:CBF-SDP-general}
        \bk(\bx) = \argmin_{\bu\in\real^m} \quad &\|\bu-\bk_\des(\bx)\|^2                               \\
                  \textup{s.t.} \quad &\dot \bH(\bx,\bu) \succeq -\balpha(\bH(\bx))\nonumber
    \end{align}
    is continuous for all $\bx$ in some neighborhood $\Ec$ of $\Cc$. Consequently, the state-feedback $\bu=\bk(\bx)$ renders set $\Cc$ forward invariant for the closed-loop system.~\hfill$\blacksquare$
\end{theorem}
Theorem~\ref{thm:safety_MCBF} guarantees safety with a continuous CBF-SDP safety filter. We would like to emphasize that the function $\balpha$ does not apply the class-$\Ke$ function~$\alpha$ on $\bH$ entry-wise but on its eigenvalues. Through spectral decomposition, we extend the exponential MCBF from~\ref{sec:EMCBF} to the more general class with class-$\Ke$ function. Note importantly from Lemma~\ref{lem:continuity_class-K} that the selection of the orthonormal eigenvectors $\{\bv_j\}_{j=1}^p$ can be independently performed for each $\bx$, making the constraint formulation for the safety filter~\eqref{eq:CBF-SDP-general} practical.

\section{Further Discussions}
In this section, we offer several analytical insights and possible refinements of the proposed MCBF framework. These include a less conservative barrier condition, a more general form of the safety filter, and a comparison between matrix- and scalar-valued representations of safe sets. While we do not provide the full technical details, we consider these discussions to be an important complement to the main results.

\subsection{Less Conservative Barrier Condition}
The barrier condition proposed in~\eqref{eq:BC-diagonalization} is not the least conservative formulation possible. By Sylvester's criterion~\cite[Thm. 7.2.5]{RAH-CRJ:12}, a matrix is positive semidefinite if and only if all of its leading principal minors\footnote{Leading principal minors are the determinants of the upper left $j\times j$ corner submatrices for all $j\in\until{p}$.} are nonnegative. This, in turn, implies that all upper left $j\times j$ corner submatrices must be positive semidefinite. Consequently, the matrix inequality in~\eqref{eq:BC-diagonalization} is equivalent to:
\begin{align*}
    \bV_{\until{j}}^\top(\bx)\bL_\bf\bH(\bx)\bV_{\until{j}}(\bx) & \succeq -\diag\big(\alpha(\lambda_1(\bx)),\dots,\alpha(\lambda_j(\bx))\big)
\end{align*}
for all $j\in\until{p}$, where the matrix $\bV_{\until{j}}(\bx)$ is the concatenation of only the orthonormal eigenvectors associated with the smallest $j$ eigenvalues. Here, the minimum eigenvalue of the left hand side can be interpreted as the worst-case rate of change among the $j$ smallest eigenvalues. Thus, we can further relax the condition by replacing the right hand side with:
\begin{equation}\label{eq:BC_multiple_LMIs}
    \bV_{\until{j}}^\top(\bx)\bL_\bf\bH(\bx)\bV_{\until{j}}(\bx) \succeq -\alpha(\lambda_j(\bx))\bI_{p\times p}
\end{equation}
while maintaining safety guarantees.

The condition~\eqref{eq:BC_multiple_LMIs} is equivalent to the barrier condition proposed in the previous work~\cite{PO-BC-LS-JC:23-auto}. However, that work formulated the control barrier function as a hierarchical optimization and did not make a connection to LMIs, which greatly reduce computational overhead during implementation. From a computational standpoint, the barrier condition~\eqref{eq:BC_multiple_LMIs} scales poorly with matrix size because it contains multiple LMIs. In contrast, the CBF~\eqref{eq:CBF-SDP-general} proposed in this paper requires a single LMI. For this reason, we omit the full technical treatment of this alternative.

\subsection{Safety Filter Cost Functions}
Beyond the choice of barrier conditions, the cost function plays a key role in the behavior of safety filters. The standard safety filter cost function encodes minimal deviations from desired control inputs. For this particular cost function, we have established the continuity of the safety filter induced by a MCBF. However, this cost function can be generalized as:
\begin{align*}
    \bk(\bx)  = \argmin_{\bu\in\real^m} \quad & J(\bx,\bu)                                  \\
             \textup{s.t.} \quad & \dot \bH(\bx,\bu) \succeq -\balpha(\bH(\bx))
\end{align*}
with some cost function $\map{J}{\real^n\times\real^m}{\real}$. In this case, continuity of the safety filter can be established if $J$ is continuous in $\bx$ and strictly convex in $\bu$ (or simply if the sub-level sets of $J$ is compact for each $\bx$). We omit the result due to the increased complexity in the proof, cf.~\cite{PO-BC-LS-JC:23-auto}.

\subsection{Schur's Complement Reformulation}

Schur's complement is often useful for reformulating matrix safety constraints as scalar ones, and vice versa. We briefly investigate here the differences between the two representations. While both representations of the safe set are equivalent, the resulting CBF constraints are not, as illustrated with the following spherical constraint example. 

Consider the control-affine system~\eqref{sys:ctrl-affine} with a simple spherical constraint of radius $R$ centered at $\bx_\obs$. There are two equivalent representations:
$$
    \|\bx-\bx_\obs\|^2\leq R^2 \iff \begin{bmatrix}
        \bI_{n\times n}     & \bx-\bx_\obs \\
        (\bx-\bx_\obs)^\top & R^2
    \end{bmatrix}\succeq \bzero
$$
On the one hand, the standard exponential CBF approach gives the following constraint:
\begin{equation*}
    -2(\bx-\bx_\obs)^\top\big(\bf(\bx)+\bg(\bx)\bu\big) \geq -c_\alpha(R^2-\|\bx-\bx_\obs\|^2),
\end{equation*}
which can be algebraically rearranged into:
\begin{multline*}
    c_\alpha R^2 - \frac{1}{c_\alpha}\|c_\alpha (\bx-\bx_\obs)+\bf(\bx)+\bg(\bx)\bu\|^2 \\\geq -\frac{1}{c_\alpha}\|\bf(\bx)+\bg(\bx)\bu\|^2.
\end{multline*}
On the other hand, the  MCBF approach from Sec.~\ref{sec:EMCBF} yields:
\begin{align*}
     & \begin{bmatrix}
           \bzero_{n\times n} & \dot \bx \\
           \dot \bx^\top      & 0
       \end{bmatrix}\succeq-c_\alpha\begin{bmatrix}
       \bI_{n\times n}     & \bx-\bx_\obs \\     (\bx-\bx_\obs)^\top & R^2
        \end{bmatrix}
    \\
     & \iff c_\alpha R^2 - \frac{1}{c_\alpha}\|\bf(\bx)+\bg(\bx)\bu+c_\alpha(\bx-\bx_\obs)\|^2 \geq 0,
\end{align*}
which differs by the extra term $(1/c_\alpha)\|\bf(\bx)+\bg(\bx)\bu\|^2$. This difference hints at a potentially useful generalization, whose implications are left for future work.

\section{Applications}
\label{sec:applications}
The MCBF formulation offers safe control solutions for a broader class of safety problems. This section demonstrates its versatility through discussions of specific applications.

\subsection{Localization via Nonlinear Least Squares}
In many robotic systems, the full state $\bx$ is not directly measured. Instead, we have sensor measurements $\by$ through an output of some function $\map{\bmm}{\real^n}{\real^p}$:
$$
\by = \bmm(\bx).
$$
One approach for computing a state estimation $\hat \bx\in\real^n$ from measurements is through solving nonlinear least square (NLS) problems:
$$
\hat\bx = \argmin_{\bx\in\real^n} \|\by-\bmm(\bx)\|^2,
$$
a method increasingly used in robotics in fields such as SLAM, pose graph optimization, and sensor fusion. One recent work~\cite{SGG-DT-BTL:24} has begun to explore how control barrier functions can be used to ensure NLS remains well-posed during robot operation.

In particular, the Hessian matrix of the cost function:
$$
\bH(\bx;\by) = \nabla^2 \|\by-\bmm(\bx)\|^2
$$
must remain positive definite as the robot navigates through its environment. This would ensure the optimization problem admits a unique solution; otherwise, the resulting estimates may be ambiguous or discontinuous. The work~\cite{SGG-DT-BTL:24} approaches this problem using standard CBFs on the smallest eigenvalue of $\bH$. As a result, additional constraints must be imposed to avoid the possibility of nonsmoothness. In contrast, our MCBF formulation handles such matrix inequalities directly and smoothly, allowing for the construction of a continuous safety filter that ensures well-posedness of the localization problem. Due to technical considerations beyond the scope of this paper, the full technical treatment of this application will appear in a separate paper.

\subsection{Collision Avoidance}
Our formulation expands the range of geometrical objects that CBFs can handle in collision avoidance tasks for robotic applications. Even when we limit the matrix $\bH$ to be linear in the state $\bx$, the set $\Cc$ in~\eqref{eq:safeset-matrix} describes a spectrahedron. This class of objects contains many familiar and useful shapes such as polyhedra, hyperspheres, cylinders, and cones. It also encompasses interesting objects such as elliptopes and objects whose shadows are ellipses with multiple foci (m-ellipses). As an example, Fig.~\ref{fig:elliptope} depicts the elliptope defined by:
\begin{equation}
\label{eq:elliptope}
    \bH(\bx) = \left\{\bx\in\real^3\middle|\begin{bmatrix}
        1 & x_1 & x_2\\ x_1 & 1 & x_3 \\ x_2 & x_3 & 1
    \end{bmatrix} \succeq 0\right\}.
\end{equation}
We hope the generalization to matrix inequalities will enable more sophisticated analytical representations of real-world obstacles.
\begin{figure}[h]
    \centering
    \includegraphics[width=0.95\linewidth]{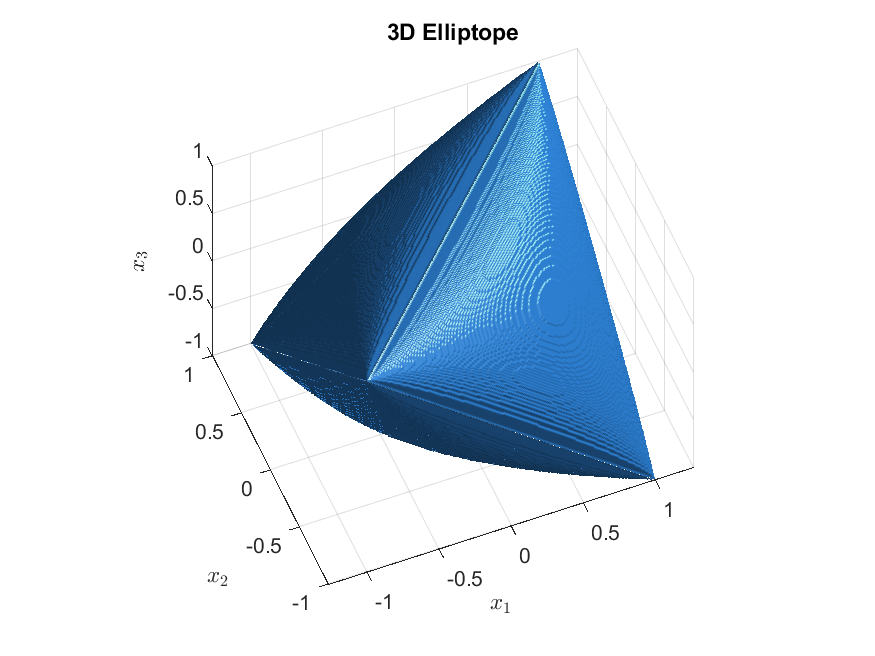}
    \caption{An example of a spectrahedron known as an \textit{elliptope}, as described by \eqref{eq:elliptope}.}
    \label{fig:elliptope}
\end{figure}

\subsection{Connectivity Maintenance}
\label{sec:connectivity}
An immediate application of MCBFs is the task of connectivity maintenance in a multi-robot system. We consider a group of $p$ robot agents. Let $\bx_i$ be the state of the $i$-th robot and $\bx$ be the aggregate of all the states evolving according to the control-affine dynamics~\eqref{sys:ctrl-affine}. One key safety constraint concerns the communication capability between robots.

We use a weighted adjacency matrix $\map{\bA}{\real^n}{\mathbb{S}^p}$ to describe connectivity between robots. Note here we only consider the case of symmetric adjacency matrices, which correspond to undirected communication graphs. In particular, the entry $A_{ij}(\bx)> 0$ if and only if robots $i$ and $j$ are connected. The Laplacian matrix is defined as $\bL(\bx)\defeq\bD(\bx)-\bA(\bx)$, where $\bD_{ii}(\bx)\defeq \sum_{j\in \until{p}} \mathbf{A}_{ij}(\bx)$ is the degree matrix.  By construction, the smallest eigenvalue $\phi_1(\bL(\bx))\equiv0$ is always zero, with an eigenvector being a vector of ones $\mathbf{1}_p$. More importantly, the second smallest eigenvalue $\phi_2(\bL(\bx))$ is useful for describing the connectivity of the robot network. A network is connected (i.e., there is a communication \textit{path} between any pair of robots) if and only if $\phi_2(\bL(\bx))>0$.

Connectivity can be enforced using a MCBF. By proposing:
\begin{equation}\label{eq:connectivity_MCBF}
    \bH(\bx) = \bL(\bx)+\frac{\varepsilon}{p}\mathbf{1}_p\mathbf{1}_p^\top-\varepsilon\bI_{p\times p}
\end{equation}
with $\varepsilon\geq 0$, the set $\Cc$ comprises only states $\bx$ for which the network is connected. If $\bH$ is an MCBF, we can design a safety filter to maintain the robot network connectivity. The CBF approach enables a simple integration of the connectivity constraint that minimally interferes with other robot tasks. In addition, our matrix-based formulation avoids the nonsmooth issues associated with eigenvalues (see \cite{PO-BC-LS-JC:23-auto} for the discussion on this issue).

\section{Multi-UAV System Demonstration}

\begin{figure*}[t]
\centering
\includegraphics[width=\linewidth]{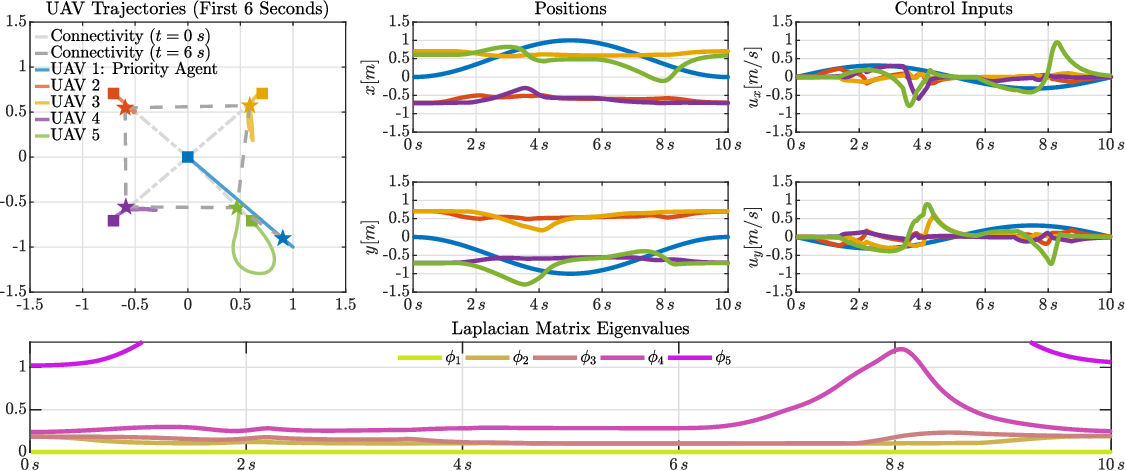}
\caption{Simulation--Five-Agent Connectivity. \textbf{[Top Left]} The simulated trajectories of all five UAVs during the first six seconds of the connectivity experiment. The priority agent (UAV 1) follows its time-varying reference while the other UAVs respond to maintain connectivity while avoiding collisions. \textbf{[Top Middle]} The 2-D position of all five UAVs over the course of the 10-second connectivity experiment. All agents deviate from their nominal starting positions to maintain forward invariance of the prescribed safe set. \textbf{[Top Right]} The 2-D control input of all UAVs. The control signals are continuous and physically achievable. \textbf{[Bottom]} The eigenvalues of the Laplacian matrix. Four of the five eigenvalues are positive, confirming that all agents maintained connectivity.}
\label{fig:simulation}
\end{figure*}

\subsection{Scenario Description}

To demonstrate the novel utility of our proposed MCBF-based safety filter, we further explore the application of connectivity maintenance, as introduced in Section\,\ref{sec:connectivity}, by constructing a two-dimensional multi-agent scenario which can be applied to a swarm of five quadrotor \textit{unmanned aerial vehicles} (UAVs). In this scenario, we first command each UAV to track a nominal reference\footnote{Note that the reference position for UAV 1 (i.e. $\bx_{1,\des}$) is a time-varying function, while all other reference positions are static.} defined by:
\begin{align*}
   &\bx_{1,\mathrm{d}}(t) = \left(1 - \cos\left(\frac{\pi}{5}t\right)\right)
    \begin{bmatrix}
        \frac{1}{2}\\\frac{-1}{2}
    \end{bmatrix},~\bx_{2,\mathrm{d}} = 
    \begin{bmatrix}
        \frac{-\sqrt2}{2}\\\frac{\sqrt2}{2}
    \end{bmatrix},\\
    &\bx_{3,\mathrm{d}} = 
    \begin{bmatrix}
        \frac{\sqrt2}{2}\\\frac{\sqrt2}{2}
    \end{bmatrix},~\bx_{4,\mathrm{d}} = 
    \begin{bmatrix}
        \frac{-\sqrt2}{2}\\\frac{-\sqrt2}{2}
    \end{bmatrix},~ 
    \bx_{5,\mathrm{d}} = 
    \begin{bmatrix}
        \frac{5\sqrt2-1}{10}\\\frac{-\sqrt2}{2}
    \end{bmatrix},
\end{align*}

\noindent where $\bx_\mathrm{i,\mathrm{d}}\in\R^2$ for $i\in\until{5}$ represents the desired position of the $i$-th UAV. Each reference is tracked with the nominal proportional tracking controller $\bk_{\des}:\R^n\rightarrow\R^m$ defined as:

\begin{equation}
\label{eq:p-control}
\bu_{i,\des}=\bk_{\des}(\bx_i) = k(\bx_{i,\des} - \bx_i) + \dot{\bx}_{i,\des},
\end{equation}

\noindent with proportional gain $k\in\realpos$. This simple control law generates a nominal input $\bu_{i,\des}$ that yields exponentially-stable lag-free tracking for systems with fully-actuated single-integrator dynamics. 

Next, we pass these nominal commands through the MCBF-based connectivity maintenance safety filter, as proposed in Section\,\ref{sec:connectivity}. Connectivity between UAVs is captured by a proximity-based adjacency matrix $\bA$ with entries:
$$
A_{ij}(\bx)= \begin{cases}\exp(1-\|x_i-x_j\|^2/R^2)-1 & \text{if}~\|x_i-x_j\| \leq R \\
0 & \text{otherwise}\end{cases},
$$
where $R\in\R_{>0}$ denotes the maximum communication range between agents. For the purposes of this demonstration, we set $R=1.3\;m$, which yields a network that is initially connected at $t=0$ and would become disconnected shortly thereafter without safety filtering. Using this adjacency matrix, we construct the MCBF in~\eqref{eq:connectivity_MCBF} with $\varepsilon=0.1$ for adequate robustness. The MCBF condition~\eqref{eq:CBC-matrix} can be verified since the overall system can be modeled using single integrators, allowing for omnidirectional UAV motion.

In addition to the connectivity constraint, we enforce pairwise collision avoidance between UAVs using standard scalar-valued CBFs. Between UAV $i$ and $j$, we have:
$$
h_\text{col}^{ij}(\bx_i,\bx_j) =\|\bx_i-\bx_j\|^2-4r_\text{agent}^2,
$$
where $r_\text{agent}\in\realpos$ is the prescribed collision radius of each agent, which is set to $r_\text{agent}=0.25\;m$. Furthermore, we designate a priority agent (UAV 1) whose control is fixed at its desired value, i.e., $\bu_1 = \bu_{1,\mathrm{d}}$.

Finally, the safety filter is formulated with the CBF-SDP:
\begin{subequations}
\label{eq:sdp}
\begin{align}
    \bk(\bx) = \argmin_{\bu\in\real^m} \quad & \sum_{i=1}^5 \|\bu_{i,\des}-\bu_i\|^2,\\
    \text{s.t.} \quad & \dot \bH(\bx,\bu) \succeq -\balpha_1(\bH(\bx)),\\
    & \dot{h}_{\text{col}}^{ij}(\bx,\bu) \geq -\alpha_2(h_{\text{col}}^{ij}(\bx)), \\
    & \qquad \qquad  \forall i< j,~\text{with}~i,j\in\until{5}\nonumber\\
    & \bu_1 = \bu_{1,\mathrm{d}},
\end{align}
\end{subequations}
where $\bx=[\bx_1^\top,...,\bx_5^\top]^\top$ and $\bu=[\bu_1^\top,...,\bu_5^\top]^\top$. This centralized safety filter computes the minimum modification to the nominal control action for each UAV that guarantees forward invariance of all prescribed safe sets.

\subsection{Simulation Results}

Before evaluating this proposed scenario on a physical UAV system, we construct an ideal simulation environment to establish a nominal performance baseline. Our simulation models a multi-agent system operating in 2-D space, where each agent follows the single integrator dynamics: $\dot \bx_i = \bu_i$ with $\bx_i\in\real^2$. The resultant behavior is depicted in Fig.~\ref{fig:simulation}.

Inspecting the figure, we see that all five simulated trajectories remain connected throughout the entire simulation. As UAV 1 (priority agent) leaves the starting formation, UAVs 2-4 respond by abandoning their nominal positions, following UAV 1 to stay connected. Meanwhile, UAV 5 is initially forced to avoid a collision with UAV 1. Subsequently, it returns to its nominal position and serves as a relay between UAV 1 and UAVs 2-4. The control signals produced by \eqref{eq:sdp} are continuous, and the nontrivial eigenvalues of the Laplacian matrix remain positive. Fig.~\ref{fig:simulation} (bottom) shows that the eigenvalues $\phi_2$ and $\phi_3$ merge several times during the simulation. For such behavior, methods that treat an individual eigenvalue as a scalar-valued CBF~\cite{BC-LS:20} would typically suffer from high-frequency chatter in the control signals (see Fig.~\ref{fig:chatter}) due to the discontinuity of the associated barrier condition, cf. \cite{PO-BC-LS-JC:23-auto}. However, our novel MCBF approach produces a chatter-free control input (see Fig.~\ref{fig:simulation} (top right)).

\begin{figure}[h]
    \centering
    \includegraphics[width=\linewidth]{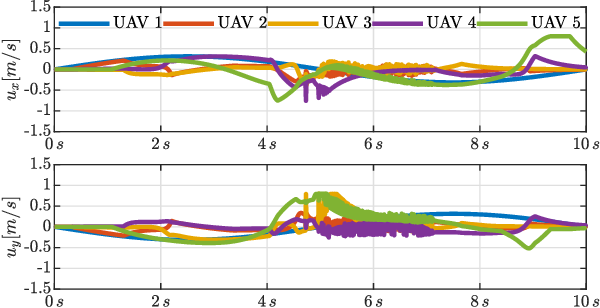}
    \caption{An example of control signal chattering when using eigenvalues as scalar-valued CBFs~\cite{BC-LS:20} in the simulated connectivity scenario. The eignenvalues merge at various times, and the discontinuity of the associated barrier condition produces chatter which could destabilize the physical system.}
    \label{fig:chatter}
\end{figure}

In addition to these pre-computed simulation results, we have provided a GitHub repository with Python code to demonstrate the use of a MCBF-based safety filter for maintaining multi-agent network connectivity. The repository is available at the following URL: \url{https://github.com/pioong/connectivity_MCBF}. Within this repository is an interactive user-controlled simulation, implemented using a \texttt{pygame} environment, in which the user can arbitrarily prescribe the priority agent (and its corresponding position reference) in real-time. Additional details are provided in the \texttt{README} file included in the repository.

\begin{figure*}[t]
\centering
\includegraphics[width=\linewidth]{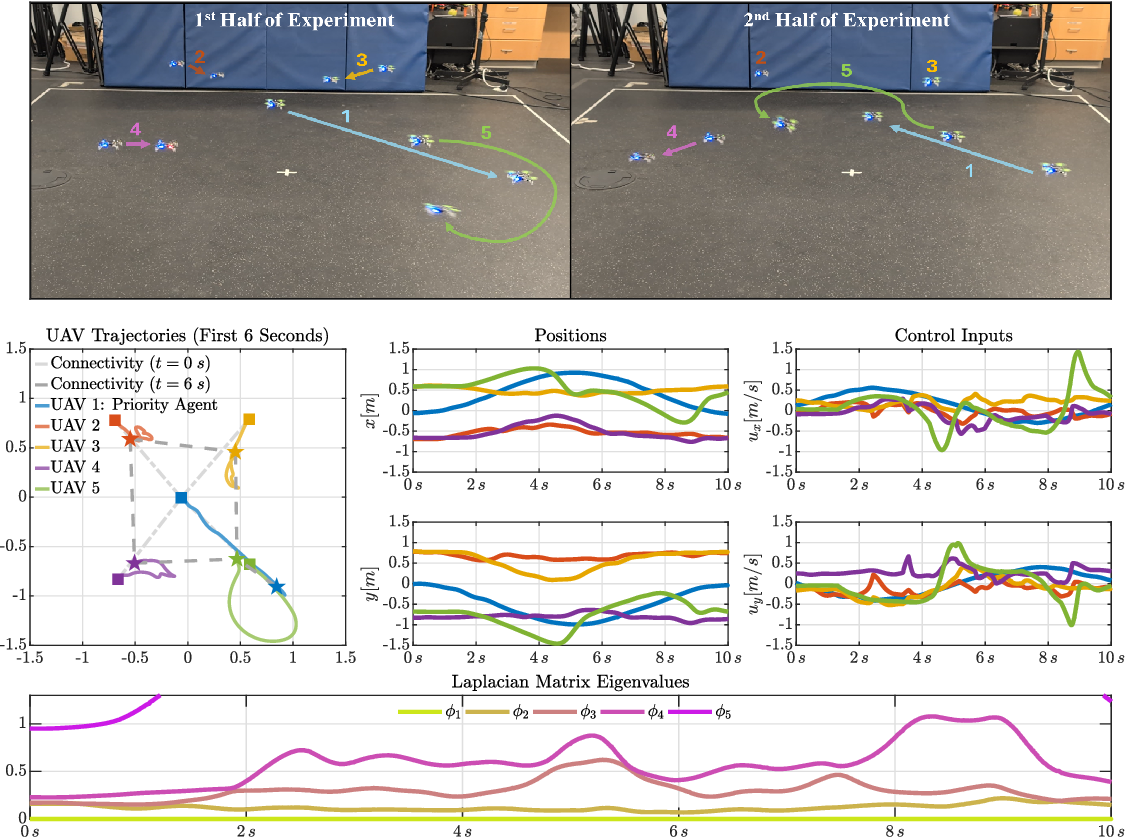}
\caption{Experiment--Five-Agent Connectivity. \textbf{[Top]} Experiment timelapse photo showing the physical positions and approximate trajectories at various times throughout the experiment. The first photo depicts $t\in[0,5]$, while the second photo depicts $t\in[6,9]$. \textbf{[Middle Left]} The measured UAV trajectories during the first six seconds of the experiment. Qualitatively, the overall behavior matches the simulation results. \textbf{[Middle Center]} Measured UAV positions. These values closely align with the simulated expectation. \textbf{[Middle Right]} Computed UAV control inputs. The control signals deviate marginally from the simulation, which is expected as the real UAVs are subject to physical disturbances, sensor noise, and model uncertainty. \textbf{[Bottom]} Laplacian matrix eigenvalues. As with the simulation, all agents maintained connectivity.}
\label{fig:experiment}
\end{figure*}

\subsection{Experiment Results}

Following our successful simulated demonstrations, we physically implemented the aforementioned connectivity maintenance scenario on a swarm of Crazyflie 2.1+ UAVs. We used an OptiTrack motion capture system to perform localization at 240~Hz, and we executed the control laws in \eqref{eq:p-control} and \eqref{eq:sdp} on a centralized offboard PC \footnote{Computations were performed on a custom PC with an AMD Ryzen 9 9950x CPU running an Ubuntu 24.04 operating system.}. Using a \texttt{clarabel} SDP solver, we achieved optimal safety filter solutions on the order of 1-2~milliseconds, enabling the synthesis of safe 2-D velocity commands at the state update frequency of 240~Hz. Simultaneously, we generated a vertical velocity command through a decoupled altitude tracking controller, leveraging the same control law as in \eqref{eq:p-control}. After converting these 3-D velocities to acceleration commands (by backstepping through an additional proportional feedback loop), we radioed them to each UAV in real-time, where they were tracked at 1000~Hz by an onboard low-level quaternion-based attitude controller. The results of this experimental demonstration can be seen in Fig.~\ref{fig:experiment}, and extended video footage is available at \url{https://youtu.be/CuSsAjOmPik}.

When we examine the experimental data, our first key observation is that the results closely parallel those of the simulation. This is especially apparent when inspecting the position plots in Fig.~\ref{fig:experiment} (middle left \& middle center) which bear a strong resemblance to those in Fig.~\ref{fig:simulation}, despite the obvious limitations of using a single-integrator model to represent UAV flight dynamics. Another important observation is that the control signal in Fig.~\ref{fig:experiment} (middle right) is continuous, physically reasonable, and chatter-free, which further highlights the benefit of our MCBF-based architecture, when compared to existing methods. Lastly, we note the eigenvalues of the Laplacian matrix. As in the simulation, these eigenvalues merge without creating problems in the control signals, and they stay above zero for the duration of the experiment, confirming that connectivity is maintained.

\section{Conclusion}
This paper generalizes the control barrier function framework to accommodate matrix-valued functions. We have considered both the semidefinite and indefinite constraints on these matrix functions. We established matrix barrier conditions for both autonomous and control systems, and proved the continuity of the CBF-SDP controller derived from the proposed MCBFs. The framework was demonstrated through a multi-robot network connectivity maintenance task, validated in both simulation and quadrotor experiments. Our future work focuses on both theoretical development and practical applications, such as identifying sufficient conditions for local Lipschitz continuity of the CBF-SDP controllers, and exploring applications beyond connectivity maintenance, including Boolean compositions of scalar-valued CBFs.

\bibliographystyle{ieeetr}
\bibliography{
    bib/alias,
    bib/PO,
    bib/main-Pio
}

\vspace*{-7ex}

\begin{IEEEbiography}[{\includegraphics[width=1in,height=1.25in,clip,keepaspectratio]{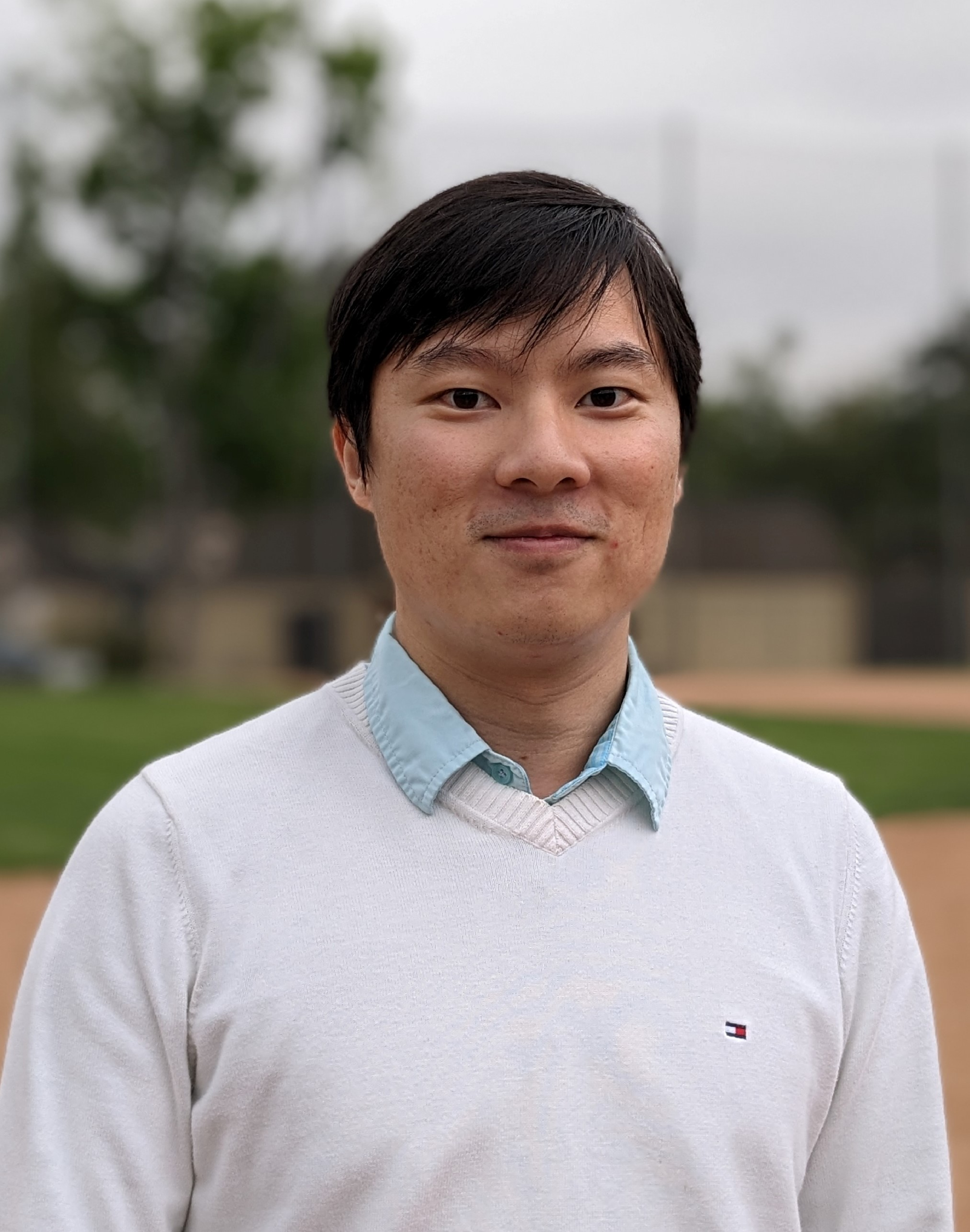}}]{Pio
    Ong} (Member, IEEE)
received the B.S. degree in Aerospace Engineering from University of California, San Diego (UCSD) in 2012, a M.S. degree in Astronautical Engineering from University of Southern California (USC) in 2013. During Fall and Winter of 2014, he worked at Space Exploration Technologies Corp (SpaceX). He received the Ph.D. degree in Engineering Sciences (Aerospace Engineering) from UCSD in 2022 with Professor Jorge Cort\'{e}s as his advisor. Currently, Pio is a Postdoctoral Scholar at California Institute of Technology (Caltech), under the supervision of Professor Aaron D. Ames. His current research interests include control barrier functions, event-triggered control, resilient autonomy, and applications to multi-agent and aerospace systems.
\end{IEEEbiography}

\vspace*{-7ex}

\begin{IEEEbiography}
[{\includegraphics[width=1in,height=1.25in,clip,keepaspectratio]{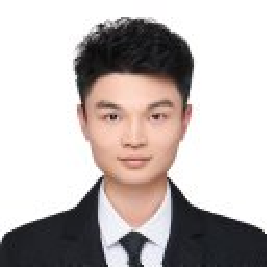}}]{Yicheng Xu} (Graduate Student Member, IEEE) received the B.S. degree in automation from Southeast University, Nanjing, China in 2020 and the M.S. degree in mechanical engineering from University of California, Irvine in 2021. He has been a Ph.D. student since 2021. His research interests include distributed control of multiagent systems, event-triggered control and anti-windup control.
\end{IEEEbiography}

\vspace*{-7ex}

\begin{IEEEbiography}
[{\includegraphics[width=1in,height=1.25in,clip,keepaspectratio]{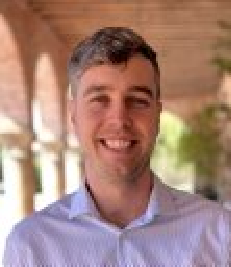}}]{Ryan M. Bena} (Member, IEEE) received the B.S. degree in mechanical engineering from the University of California, Berkeley, CA, USA in 2012 and the M.S. degree in aerospace engineering from the University of Southern California (USC), Los Angeles, CA, USA, in 2018. From 2012 to 2019, he was an Aerospace Engineer with the US Air Force. He received his Ph.D. degree in aerospace engineering from USC in 2024, advised by Professor Quan Nguyen. He is currently a Postdoctoral Scholar at the California Institute of Technology under Professor Aaron D. Ames. His research interests include safety-critical controller design and analysis, collision avoidance for aerial robotics, and optimization-based control techniques. 
\end{IEEEbiography}

\vspace*{-7ex}

\begin{IEEEbiography}
[{\includegraphics[width=1in,height=1.25in,clip,keepaspectratio]{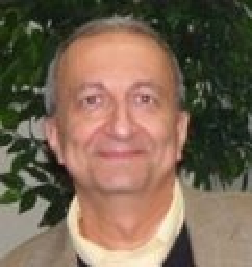}}]{Faryar Jabbari} (Life Senior Member, IEEE) is currently on the faculty of the Mechanical and Aerospace Engineering Department, University of California, Irvine, CA, USA. His research focuses on control theory and its applications. He was an Associate Editor for Automatica and IEEE Transactions on Automatic Control, and also the Program Chair for ACC-11 and CDC-09, and the General Chair for CDC-14.
\end{IEEEbiography}

\vspace*{-7ex}

\begin{IEEEbiography}[{\includegraphics[width=1in,height=1.25in,clip,keepaspectratio]{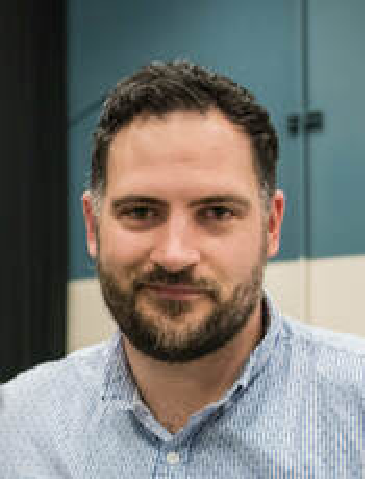}}]{Aaron D. Ames}
(Fellow, IEEE) received the joint B.S. degree in mechanical engineering and B.A. degree in mathematics from the University of St. Thomas, St Paul, MN, USA, in 2001, and the joint M.A. degree in mathematics and the Ph.D. degree in electrical engineering and computer sciences from the University of California (UC), Berkeley, Berkeley, CA USA, in 2006. He was an Associate Professor with Georgia Tech, Woodruff School of Mechanical Engineering, Atlanta, GA, USA, and the School of Electrical \& Computer Engineering, Atlanta. He is the Bren Professor of mechanical and civil engineering and control and dynamical systems with Caltech. He was a Postdoctoral Scholar in control and dynamical systems with Caltech from 2006 to 2008, and began his faculty career with Texas A\&M University, Texas, TX, USA, in 2008. His research interests include the areas of robotics, nonlinear, safety-critical control, and hybrid systems, with a special focus on applications to dynamic robots, both formally and through experimental validation. At UC Berkeley, he was the recipient of the 2005 Leon O. Chua Award for achievement in nonlinear science and the 2006 Bernard Friedman Memorial Prize in Applied Mathematics, NSF CAREER Award in 2010, 2015 Donald P. Eckman Award, and 2019 IEEE CSS Antonio Ruberti Young Researcher Prize.
\end{IEEEbiography}

\end{document}